\documentclass[11pt,reqno]{article}
\usepackage{amssymb,amscd,amsmath,amsthm}
\usepackage{graphicx}
\def\natural{\mathbf{N}}

\newtheorem{theorem}{Theorem}[section]
\newtheorem{proposition}[theorem]{Proposition}
\newtheorem{corollary}[theorem]{Corollary}
\newtheorem{lemma}[theorem]{Lemma}
\newtheorem{definition}[theorem]{Definition}

\newtheorem{remark}[theorem]{Remark}
\def\cb{{\mathcal B}}
\def\ç{{\mathcal C}}

\def\ce{{\mathcal E}}

\def\cw{{\mathcal W}}

\def\bc{{\mathbb C}}
\def\bh{{\mathbb H}}

\def\bn{{\mathbb N}}

\def\br{{\mathbb R}}

\def\frak{\mathfrak}

\def\ga{{\frak A}}

\def\a{\alpha}
\def\b{\beta}

\def\tr{{\rm Tr}}
\def\L{\Lambda}
\def\G{\Gamma}
\def\ce{\mathcal E}

\def\ffi{\varphi}

\def\Tr{\mathrm{Tr}}
\def\<{\langle}
\def\>{\rangle}

\def\1{\mathbf{1}}

\def\ve{\varepsilon}
\def\cw{\cal W}

\def\cal{\mathcal}

\def\s{\sigma}
\def\bh{\mathbf{h}}
\def\bs{\mathbf{s}}

\def\id{{\bf 1}\!\!{\rm I}}

\begin{document}
\begin{center}
{\Large {\bf  Phase transitions for Quantum Markov Chains associated with Ising type models on a Cayley tree}}\\[1cm]
\end{center}

\begin{center}
{\large {\sc Farrukh Mukhamedov}}\\[2mm]
\textit{ Department of Computational \& Theoretical Sciences,\\
Faculty of Science, International Islamic University Malaysia,\\
P.O. Box, 141, 25710, Kuantan, Pahang, Malaysia}\\
E-mail: {\tt far75m@yandex.ru, \ farrukh\_m@iium.edu.my}
\end{center}
\begin{center}

{\sc Abdessatar Barhoumi}\\
\textit{Department of Mathematics\\
Nabeul Preparatory Engineering Institute\\
Campus Universitairy - Mrezgua - 8000 Nabeul,\\
Carthage University, Tunisia}\\
E-mail: {\tt abdessatar.barhoumi@ipein.rnu.tn}\\
\end{center}

\begin{center}
{\sc Abdessatar Souissi}\\
\textit{
Department of Mathematics,\\
Marsa Preparatory Institute for Scientific and Technical Studies\\
Carthage University, Tunisia}\\
E-mail: {\tt s.abdessatar@hotmail.fr}\\
\end{center}

\begin{abstract}
The main aim of the present paper is to prove the existence of a
phase transition in quantum Markov chain (QMC) scheme for the
Ising type models on a Cayley tree. Note that this kind of models
do not have one-dimensional analogous, i.e. the considered model
persists only on trees. In this paper, we provide a more general
construction of forward QMC. In that construction, a QMC is
defined as a weak limit of finite volume states with boundary
conditions, i.e. QMC depends on the boundary conditions. Our main
result states the existence of a phase transition for the Ising
model with competing interactions on a Cayley tree of order two.
By the phase transition we mean the existence of two distinct QMC
which are not quasi-equivalent and their supports do not overlap.
We also study some algebraic property of the disordered phase of
the model, which is a new phenomena even in a classical setting.

\end{abstract}

\section{Introduction }\label{intr}

One of the basic open problems in quantum probability is the
construction of a theory of quantum Markov fields, that are
quantum processes with multi-dimensional index set. This program
concerns the generalization of the theory of Markov fields (see
\cite{D},\cite{Geor})) to a non-commutative setting, naturally
arising in quantum statistical mechanics and quantum field theory.

The quantum analogues of Markov chains were first constructed in
\cite{[Ac74f]}, where the notion of quantum Markov chain (QMC) on
infinite tensor product algebras was introduced. The reader is
referred to \cite{AW,fannes2,ILW,Mat,OP} and the references cited
therein, for recent developments of the theory and the applications.

The main aim of the present paper is to prove the existence of
phase transitions for a class of quantum Markov chains associated
with Ising type models on a Cayley tree. This paper deals with two
problems: The construction of quantum Markov fields on homogeneous
(Cayley) trees; the existence of a phase transition for special
models of such fields. Both problems are non-trivial and, to a
large extent, open. In fact, even if several definitions of
quantum Markov fields on trees (and more generally on graphs) have
been proposed, a really satisfactory, general theory is still
missing and physically interesting examples of such fields in
dimension $d\geq2$ are very few. In this paper, from QMC
perspective, we are going to establish the phase transition for
the classical Ising model with competing interactions on a Cayley
tree. Note that the phase transition notion is based on the
quasi-equivalence of QMC which differs from the classical one (in
the classical setting, to establish the phase transition for the
considered model, it is sufficient to prove the existence of at
least two different solutions of associated renormalized equations
(see \cite{GPW,Roz2})). Therefore, such a phase transition is
purely noncommutative, and even for classical models, to check the
existence of the phase transition is not a trivial problem (we
point out that the quasi-equivalence of product states, which
correspond to the classical models without interactions, was
considered in \cite{PS}).

We notice that first attempts to construct a quantum analogue of
classical Markov fields have been done in
\cite{[AcFi03]}-\cite{[AcFi01b]},\cite{AcLi,GZ,FM,[Liebs99]}. In
these papers the notion of {\it quantum Markov state}, introduced
in \cite{[AcFr80]}, extended to fields as a sub-class of the
quantum Markov chains. In \cite{[AcFiMu07]} a more general
definition of quantum Markov states and chains, including all the
presently known examples, have been extended. Note that in the
mentioned papers quantum Markov fields were considered over
multidimensional integer lattice $\mathbb{Z}$. This lattice has
so-called amenability property. Moreover, analytical solutions
(for example, critical temperature)does not exist on such lattice.
But investigations of phase transitions of spin models on
hierarchical lattices showed that there are exact calculations of
various physical quantities (see for example, \cite{Bax,Per}).
Such studies on the hierarchical lattices begun with the
development of the Migdal-Kadanoff renormalization group method
where the lattices emerged as approximants of the ordinary crystal
ones.  On the other hand, the study of exactly solved models
deserves some general interest in statistical mechanics
\cite{Bax}. Therefore, it is natural to investigate quantum Markov
fields over hierarchical lattices.  For example, a Cayley tree is
the simplest hierarchical lattice with non-amenable graph
structure \cite{Ost}. This means that the ratio of the number of
boundary sites to the number of interior sites of the Cayley tree
tends to a nonzero constant in the thermodynamic limit of a large
system. Nevertheless, the Cayley tree is not a realistic lattice,
however, its amazing topology makes the exact calculations of
various quantities possible. First attempts to investigate QMC
over such trees was done in \cite{aklt}, such studies were related
to the investigation of thermodynamic limit of valence-bond-solid
models on a Cayley tree \cite{fannes} (see also \cite{AOM}).

The phase transition phenomena is crucial for quantum models over
multi dimensional lattices \cite{BCS},\cite{FS},\cite{Sach,BR2}.
In \cite{ArE} it was considered quantum phase transition for the
two-dimensional Ising model using $C^*$-algebra approach.  In
\cite{fannes} the VBS-model was considered on the Cayley tree. It
was established the existence of  the phase transition, for the
model in terms of finitely correlated states, which describes
ground states of the model. Note that more general structure of
finitely correlated states was studied in \cite{fannes2}. We
stress that finitely correlated states can be considered as
quantum Markov chains. In \cite{GM2004,Mukh04,M2000,MR1,MR2}
noncommutative extensions of classical Markov fields, associated
with Ising and Potts models on a Cayley tree, were investigated.
In the classical case, Markov fields on trees are also considered
in \cite{[Pr],Roz2},\cite{{Spa}}-\cite{[Za85]}.

There are several methods to investigate phase transition from
mathematical point of view. Roughly speaking, for a given
Hamiltonian on a quasi local algebra, to establish the existence of
a phase transition it is necessary to find at least two different
KMS-states associated with a model (see for details \cite{BR2}). In
\cite{fannes} it was proposed to study the phase transition in the
class of finitely correlated states. In the present paper, for a
given Hamiltonian we provide a more general construction (than
\cite{AMSa0,AMSa}) of QMC associated with the Hamiltonian. Namely,
in this construction, the Hamiltonian exhibits nearest-neighbor and
next-nearest-neighbor interactions (in the previous papers
\cite{AMSa} the Hamiltonian contained only nearest-neighbor
interactions), and the corresponding QMC is defined as a weak limit
of finite volume states (which depend on the Hamiltonian) with
boundary conditions, i.e. QMC depends on the boundary conditions. We
remark that in this construction, one can observe some similarities
with Gibbs measures. We stress that all previous considered examples
(in the literature) of QMC are related to Hamiltonians with
nearest-neighbor interactions. A main aim of the present paper is to
prove the existence of the phase transition in the class of QMC when
the Hamiltonian contains both kinds (nearest-neighbor and
next-nearest-neighbor) of interactions at the same time. Note that
this kind of models do not have one-dimensional analogous, i.e. the
considered model persists only on trees. In classical setting, this
kind of model was called the Ising model with competing
interactions, and has been rigorously investigated in many papers
(see for example \cite{GPW,MR1,MR2,OMM,Roz2,RRah}. Our main result
is the following theorem

\begin{theorem}\label{Main}
For the Ising model with competing interactions \eqref{1Kxy1},
\eqref{1Hxy1}, $J>0$, $\b>0$ on the Cayley tree of order two, the
following statements hold:
\begin{itemize}
\item[(i)] if $\Delta(\theta)\leq0$, then there is a unique QMC;

\item[(ii)] if $\Delta(\theta)>0$, then there occurs a phase
transition.
\end{itemize}
Here $\Delta(\theta)=\theta^{J}(\theta^{2}-3)-2\theta$,
$\theta=e^{2\beta}$.
\end{theorem}

By the phase transition we mean the existence of two distinct QMC
for the given family of interaction operators $\{K_{<x,y>}\}$,
$\{L_{>x,y<}\}$ (see \eqref{1Kxy1}, \eqref{1Hxy1}). Moreover, these
states should be not quasi-equivalent and their supports do not
overlap. Note that in our earlier papers (see \cite{AMSa2,AMSa3}) we
have proved only non quasi-equivalence of the states. In this paper,
we additionally prove that the corresponding states do not have
overlapping supports. Hence, the main result of the present paper
recover a main result of \cite{AMSa3} as a particular case ($J=0$).
To prove the main result of the paper, we first establish that the
model exhibits three translation-invariant QMC $\varphi_\a$,
$\varphi_1$ and $\varphi_2$, and we study several properties of the
states $\varphi_1$ and $\varphi_2$.

We notice that the state $\varphi_\a$ corresponds to the
disordered phase of the model. In \cite{B90,Iof} it was
established that the disordered phase of the Ising model on the
Cayley tree of order $k\geq 2$ is extremal if and only if
$\theta<1/\sqrt{k}$. But for the Ising model with competing
interactions this kind of result still is unknown. In the present
paper, we find sheds some light into this question. Namely, we
will prove the following result.

\begin{theorem}\label{main2}
The states $\varphi_\a$ and $\varphi_1$ are not quasi-equivalent.
\end{theorem}

This result shows how the states relate to each other, which is
even a new phenomena in the classical setting.

Let us outline the organization of the paper. After preliminary
information (see Section 2), in Section 3 we provide a general
construction of quantum Markov chains on Cayley tree. This
construction is more general than considered in \cite{AMSa2}.
Moreover, in this section we give the definition of the phase
transition. Using the provided construction, in Section 4 we
consider the Ising model with competing interactions on the Cayley
tree of order two. Section 5 is devoted to the existence of the
three translation-invariant QMC $\varphi_\a$, $\varphi_1$ and
$\varphi_1$ corresponding to the model. Section 6 contains the
proof of Theorem \ref{Main}, namely, we first prove that states
$\varphi_1$ and $\varphi_2$ do not have overlapping supports, and
they are not quasi-equivalent. We stress that the states
$\varphi_1$ and $\varphi_2$ are not product states, and therefore,
their non quasi equivalence is independent of interest from
operator algebras point of view. In fact, in the case of product
states, many papers were devoted to the study of quasi-equivalence
\cite{BP,Mat0,PS}.
 In the final Section 7, we prove Theorem \ref{main2}.

\section{Preliminaries}

Let $\Gamma^k_+ = (L,E)$ be a semi-infinite Cayley tree of order
$k\geq 1$ with the root $x^0$ (i.e. each vertex of $\Gamma^k_+$
has exactly $k+1$ edges, except for the root $x^0$, which has $k$
edges). Here $L$ is the set of vertices and $E$ is the set of
edges. The vertices $x$ and $y$ are called {\it nearest neighbors}
and they are denoted by $l=<x,y>$ if there exists an edge
connecting them. A collection of the pairs
$<x,x_1>,\dots,<x_{d-1},y>$ is called a {\it path} from the point
$x$ to the point $y$. The distance $d(x,y), x,y\in V$, on the
Cayley tree, is the length of the shortest path from $x$ to $y$.

Recall a coordinate structure in $\G^k_+$:  every vertex $x$
(except for $x^0$) of $\G^k_+$ has coordinates $(i_1,\dots,i_n)$,
here $i_m\in\{1,\dots,k\}$, $1\leq m\leq n$ and for the vertex
$x^0$ we put $(0)$.  Namely, the symbol $(0)$ constitutes level 0,
and the sites $(i_1,\dots,i_n)$ form level $n$ (i.e. $d(x^0,x)=n$)
of the lattice (see Fig. 1).

\begin{figure}
\begin{center}
\includegraphics[width=10.07cm]{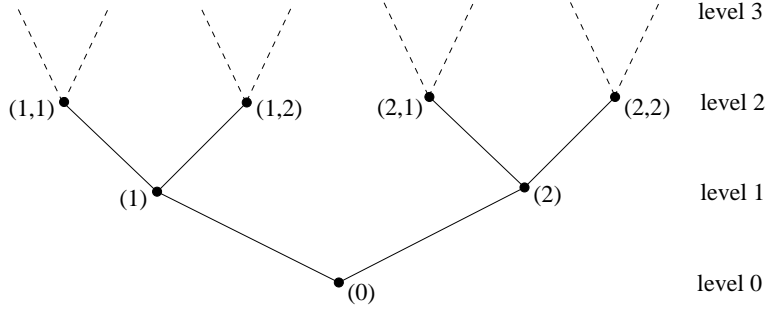}
\end{center}
\caption{The first levels of $\G_+^2$} \label{fig1}
\end{figure}

Let us set
\[
W_n = \{ x\in L \, : \, d(x,x_0) = n\} , \qquad \Lambda_n =
\bigcup_{k=0}^n W_k, \qquad  \L_{[n,m]}=\bigcup_{k=n}^mW_k, \
(n<m)
\]
\[
E_n = \big\{ <x,y> \in E \, : \, x,y \in \Lambda_n\big\}, \qquad
\Lambda_n^c = \bigcup_{k=n}^\infty W_k
\]
For $x\in \G^k_+$, $x=(i_1,\dots,i_n)$ denote
$$ S(x)=\{(x,i):\ 1\leq
i\leq k\}.
$$
Here $(x,i)$ means that $(i_1,\dots,i_n,i)$. This set is called a
set of {\it direct successors} of $x$.

Two vertices $x,y\in V$ is called {\it one level
next-nearest-neighbor  vertices} if there is a vertex $z\in V$
such that  $x,y\in S(z)$, and they are denoted by $>x,y<$. In this
case the vertices $x,z,y$ was called {\it ternary} and denoted by
$<x,z,y>$.

Let us define on $\G^k_+$ a binary operation
$\circ:\G^k_+\times\G^k_+\to\G^k_+$ as follows: for any two
elements $x=(i_1,\dots,i_n)$ and $y=(j_1,\dots,j_m)$ put
\begin{equation}\label{binar1}
x\circ
y=(i_1,\dots,i_n)\circ(j_1,\dots,j_m)=(i_1,\dots,i_n,j_1,\dots,j_m)
\end{equation}
and
\begin{equation}\label{binar2}
x\circ x^0=x^0\circ x= (i_1,\dots,i_n)\circ(0)=(i_1,\dots,i_n).
\end{equation}

By means of the defined operation $\G^k_+$ becomes a
noncommutative semigroup with a unit. Using this semigroup
structure one defines translations $\tau_g:\G^k_+\to \G^k_+$,
$g\in \G^k_+$ by
\begin{equation}\label{trans1}
\tau_g(x)=g\circ x.
\end{equation}
It is clear that $\tau_{(0)}=id$.

 The algebra of observables $\cb_x$ for any single site
$x\in L$ will be taken as the algebra $M_d$ of the complex $d\times
d$ matrices. The algebra of observables localized in the finite
volume $\L\subset L$ is then given by
$\cb_\L=\bigotimes\limits_{x\in\L}\cb_x$. As usual if
$\L^1\subset\L^2\subset L$, then $\cb_{\L^1}$ is identified as a
subalgebra of $\cb_{\L^2}$ by tensoring with unit matrices on the
sites $x\in\L^2\setminus\L^1$. Note that, in the sequel, by
$\cb_{\L,+}$ we denote the set of all positive elements of $\cb_\L$
(note that an element is positive if its spectrum is located in
$\br_+$). The full algebra $\cb_L$ of the tree is obtained in the
usual manner by an inductive limit
$$
\cb_L=\overline{\bigcup\limits_{\L_n}\cb_{\L_n}}.
$$

In what follows, by ${\cal S}({\cal B}_\L)$ we will denote the set
of all states defined on the algebra ${\cal B}_\L$.

Consider a triplet ${\cal C} \subset {\cal B} \subset {\cal A}$ of
unital $C^*$-algebras. Recall \cite{ACe} that a {\it
quasi-conditional expectation} with respect to the given triplet
is a completely positive (CP) linear map $\ce \,:\, {\cal A} \to
{\cal B}$ such that $ \ce(ca) = c \ce(a)$, for all $a\in {\cal
A},\, c \in {\cal C}$.

\begin{definition}[\cite{AOM}]\label{QMCdef}
A state $\varphi$ on ${\cal B}_L$ is called a {\it forward quantum
Markov chain (QMC)}, associated to $\{\L_n\}$, if for each
$\Lambda_n$, there exist a quasi-conditional expectation
$\ce_{\Lambda_n^c}$ with respect to the triplet
\begin{equation}\label{trplt1}
{\cal B}_{{\Lambda}_{n+1}^c}\subseteq {\cal
B}_{\Lambda_n^c}\subseteq{\cal B}_{\Lambda_{n-1}^c}
\end{equation}
and a state $ \hat\varphi_{\Lambda_n^c}\in{\cal S}({\cal
B}_{\Lambda_n^c}) $ such that for any $n\in {\mathbb N}$ one has
\begin{equation}\label{eq4.1re}
\hat\varphi_{\Lambda_n^c}| {\cal
B}_{\Lambda_{n+1}\backslash\Lambda_n} =
\hat\varphi_{\Lambda_{n+1}^c}\circ \ce_{\Lambda_{n+1}^c}| {\cal
B}_{\Lambda_{n+1}\backslash\Lambda_n}
\end{equation}
and
\begin{equation}\label{dfgqmf}
\varphi = \lim_{n\to\infty} \hat\varphi_{\Lambda_n^c}\circ
\ce_{\Lambda_n^c}\circ \ce_{\Lambda_{n-1}^c} \circ \cdots \circ
\ce_{\Lambda_1^c}
\end{equation}
in the weak-* topology.
\end{definition}

Note that \eqref{eq4.1re} is an analogue of the DRL equation from
classical statistical mechanics \cite{D, Geor}, and QMC is thus
the counterpart of the infinite-volume Gibbs measure.

\section{Construction of Quantum Markov Chains on Cayley tree}\label{dfcayley}

In this section we are going to provide a construction of a
forward quantum Markov chain which contain competing interactions.
Note that in our construction generalizes our previous works
\cite{AMSa,AMSa2}.

Let us rewrite the elements of $W_n$ in the following order, i.e.
\begin{eqnarray*}
\overrightarrow{W_n}:=\left(x^{(1)}_{W_n},x^{(2)}_{W_n},\cdots,x^{(|W_n|)}_{W_n}\right),\quad
\overleftarrow{W_n}:=\left(x^{(|W_n|)}_{W_n},x^{(|W_n|-1)}_{W_n},\cdots,
x^{(1)}_{W_n}\right).
\end{eqnarray*}
Note that $|W_n|=k^n$. Vertices
$x^{(1)}_{W_n},x^{(2)}_{W_n},\cdots,x^{(|W_n|)}_{W_n}$ of $W_n$
can be represented in terms of the coordinate system as follows
\begin{eqnarray}\label{xw}
&&x^{(1)}_{W_n}=(1,1,\cdots,1,1), \quad x^{(2)}_{W_n}=(1,1,\cdots,1,2), \ \ \cdots \quad x^{(k)}_{W_n}=(1,1,\cdots,1,k,),\\
&&x^{(k+1)}_{W_n}=(1,1,\cdots,2,1), \quad
x^{(2)}_{W_n}=(1,1,\cdots,2,2), \ \ \cdots \quad
x^{(2k)}_{W_n}=(1,1,\cdots,2,k),\nonumber
\end{eqnarray}
\[\vdots\]
\begin{eqnarray*}
&&x^{(|W_n|-k+1)}_{W_n}=(k,k,,\cdots,k,1), \
x^{(|W_n|-k+2)}_{W_n}=(k,k,\cdots,k,2),\ \ \cdots
x^{|W_n|}_{W_n}=(k,k,\cdots,k,k).
\end{eqnarray*}

Analogously, for a given vertex $x,$ we shall use the following
notation for the set of direct successors of $x$:
\begin{eqnarray*}
\overrightarrow{S(x)}:=\left((x,1),(x,2),\cdots (x,k)\right),\quad
\overleftarrow{S(x)}:=\left((x,k),(x,k-1),\cdots (x,1)\right).
\end{eqnarray*}

In what follows, by $^\circ\prod$ we denote an ordered product,
i.e.
$$
^\circ\prod_{k=1}^n  a_k=a_1a_2\cdots a_n,
$$
where elements $\{a_k\}\subset {\mathcal{B}}_L$ are multiplied in
the indicated order. This means that we are not allowed to change
this order.

Notice that each vertex $x\in L$ has interacting vertices $\{x,
(x,1),\dots,(x,k)\}$. Assume that each edges $<x,(x,i)>$
($i=1,\dots,k$) operators
$K_{<x,(x,i)>}\in\mathcal{B}_{x}\otimes\mathcal{B}_{(x, i)}$ is
assigned, respectively. Moreover, for each competing vertices
$>(x,i),(x,i+1)<$ and $<x,(x,i),(x,i+1)>$ ($i=1,\dots,k$) the
following operators are assigned:
 $$L_{>(x,i),(x,i+1)<}\in\mathcal{B}_{(x,i)}\otimes\mathcal{B}_{(x,i+1)}, \ \
 M_{(x,(x,i),(x,i+1))}\in\mathcal{B}_{x}\otimes\mathcal{B}_{(x,i)}\otimes\mathcal{B}_{(x,i+1)}.$$

We would like to define a state on $\mathcal{B}_{\Lambda_n}$ with
boundary conditions $\omega_{0}\in\mathcal{B}_{(0),+}$ and
$\{h^{x}\in\mathcal{B}_{x,+}:\ x\in L\}$. Note that the boundary
conditions have similar interpretations like classical ones (i.e.
having boundary spins parallel either all up or all down) but in
more general setting. For example, if we consider an Ising model on
a Cayley tree, then the boundary conditions for this model are
defined by functions $\{h_x\}_{x\in V}$, which in our setting
correspond to $\{\exp(h_x)\}_{x\in V}$. For this reason, we are
considering positive elements $h^{x}\in\mathcal{B}_{x,+}$ as the
boundary condition. For more information about the boundary
conditions related to classical models we refer \cite{GRR,Roz2}.

For each $n\in \natural$ denote
\begin{eqnarray}\label{K1}
&& A_{x,(x,1),\dots,(x,k)}=\big(^\circ\prod_{i=1}^k
K_{x,(x,i)}\big)\big(^\circ\prod_{i=1}^k
L_{>(x,i),(x,i+1)<}\big)\big(^\circ\prod_{i=1}^k
M_{(x,(x,i),(x,i+1))}\big),\\[2mm]
\label{K11} &&K_{[m,m+1]}:=\prod_{x\in \overrightarrow
W_{m}}A_{x,(x,1),\dots,(x,k)}, \ \ 1\leq m\leq n,\\[2mm]
 \label{K2} &&
\bh_{n}^{1/2}:=\prod_{x\in\overrightarrow
W_n}(h^{x})^{1/2}, \ \ \  \bh_n=\bh_{n}^{1/2}(\bh_{n}^{1/2})^*\\[2mm]
\label{K3}
&&{\mathbf{K}}_n:=\omega_{0}^{1/2}\prod_{m=1}^{n-1}K_{[m,m+1]}\bh_{n}^{1/2}\\[2mm]
\label{K4} && \mathcal{W}_{n]}:={\mathbf{K}}_n{\mathbf{K}}_n^{*}
\end{eqnarray}
One can see that ${\cw}_{n]}$ is positive.

In what follows, by $\tr_{\L}:\cb_L\to\cb_{\L}$ we mean normalized
partial trace (i.e. $\tr_{\L}(\id_{L})=\id_{\L}$, here
$\id_{\L}=\bigotimes\limits_{y\in \L}\id$), for any
$\Lambda\subseteq_{\text{fin}}L$. For the sake of shortness we put
$\tr_{n]} := \tr_{\Lambda_n}$.

Let us define a positive functional $\ffi^{(n)}_{w_0,\bh}$ on
$\cb_{\Lambda_n}$ by
\begin{eqnarray}\label{ffi-ff}
\ffi^{(n)}_{w_0,\bh}(a)=\tr(\cw_{n+1]}(a\otimes\id_{W_{n+1}})),
\end{eqnarray}
for every $a\in \cb_{\Lambda_n}$. Note that here, $\tr$ is a
normalized trace on ${\cal B}_L$ (i.e. $\tr(\id_{L})=1$).

To get an infinite-volume state $\ffi$ on $\cb_L$  such that
$\ffi\lceil_{\cb_{\L_n}}=\ffi^{(n)}_{w_0,\bh}$, we need to impose
some constrains to the boundary conditions $\big\{w_0,\bh\big\}$
so that the functionals $\{\ffi^{(n)}_{w_0,\bh}\}$ satisfy the
compatibility condition, i.e.
\begin{eqnarray}\label{compatibility}
\ffi^{(n+1)}_{w_0,\bh}\lceil_{\cb_{\L_n}}=\ffi^{(n)}_{w_0,\bh}.
\end{eqnarray}

In the following we need an auxiliary fact.

\begin{lemma}\label{LL}
Let $\Lambda\subseteq\Lambda'\subseteq_{fin}L$, then for any
$A\in\mathcal{B}_{\Lambda}, B\in\mathcal{B}_{\Lambda'}$ one has
$\Tr (AB)=\Tr[A\Tr_{\mathcal{B}_{\Lambda}}(B)]$.
\end{lemma}

\begin{theorem}\label{compa} Assume that for every $x\in L$ and triple $\{x,(x,i),(x,i+1)\}$ $(i=1,\dots,k-1$) the operators
$K_{<x,(x,i)>}$, $L_{>(x,i),(x,i+1)<}$, $M_{(x,(x,i),(x,i+1))}$
are given as above. Let the boundary conditions $w_{0}\in {\cal
B}_{(0),+}$ and ${\bh}=\{h_x\in {\cal B}_{x,+}\}_{x\in L}$ satisfy
the following conditions:
\begin{eqnarray}\label{eq1}
&&\Tr(\omega_{0}h^{(0)})=1, \\
\label{eq2}
&&\Tr_{x]}\big({A_{x,(x,1),\dots,(x,k)}}^\circ\prod_{i=1}^k
h^{(x,i)}A_{x,(x,1),\dots,(x,k)}^{*}\big)=h^x, \ \ \textrm{for
every}\ \ x\in L,
\end{eqnarray}
where as before $A_{x,(x,1),\dots,(x,k)}$ is given by \eqref{K1}.
Then the functionals $\{\ffi^{(n)}_{w_0,\bh}\}$ satisfy the
compatibility condition \eqref{compatibility}. Moreover, there is
a unique forward quantum Markov chain $\ffi_{w_0,{\bh}}$ on
$\cb_L$ such that
$\ffi_{w_0,{\bh}}=w-\lim\limits_{n\to\infty}\ffi^{(n)}_{w_0,\bh}$.
\end{theorem}

\begin{proof}
We first show that a sequence $\{\cw_{n]}\}$ is {\it projective}
with respect to $\tr_{n]}$, i.e.
\begin{equation}\label{pro}
\tr_{n-1]}(\cw_{n]})=\cw_{n-1]}, \ \ \forall n\in\bn.
\end{equation}
It is known \cite{[AcFr80]} that the projectivity implies the
compatibility condition.

Now let us check the equality \eqref{pro}. From
\eqref{K1}-\eqref{K4} one has
\begin{eqnarray*}
\cw_{n]}&=&
w_0^{1/2}\bigg(\prod_{m=1}^{n-1}K_{[m-1,m]}\bigg)K_{[n-1,n]}\bh_{n}K_{[n-1,n]}^*
\bigg(\prod_{m=1}^{n-1}K_{[m-1,m]}\bigg)^* w_0^{1/2}.
\end{eqnarray*}

One can see that for different $x$ and $x'$ taken from $W_{n-1}$
the algebras $\cb_{x\cup S(x)}$ and  $\cb_{x'\cup S(x')}$ are
commuting, therefore from \eqref{K2} one finds
\begin{eqnarray*}
K_{[n-1,n]}\bh_{n}K_{[n-1,n]}^*= \prod_{x\in
\overrightarrow{W}_{n-1}}\big({A_{x,(x,1),\dots,(x,k)}}^\circ\prod_{i=1}^k
h^{(x,i)}A_{x,(x,1),\dots,(x,k)}^{*}\big)
\end{eqnarray*}

Hence, from the last equality with \eqref{eq2} we get
\begin{eqnarray*}
\tr_{n-1]}(\cw_{n]})&=&
w_0^{1/2}\bigg(\prod_{m=1}^{n-1}K_{[m-1,m]}\bigg)\\
&&\times \prod_{x\in
\overrightarrow{W}_{n-1}}\tr_{x]}\bigg({A_{x,(x,1),\dots,(x,k)}}^\circ\prod_{i=1}^k
h^{(x,i)}A_{x,(x,1),\dots,(x,k)}^{*}\bigg)\\
&&\times\bigg(\prod_{m=1}^{n-1}K_{[m-1,m]}\bigg)^* w_0^{1/2}
\\
&=&w_0^{1/2}\bigg(\prod_{m=1}^{n-1}K_{[m-1,m]}\bigg)\prod_{x\in
\overrightarrow{W}_{n-1}}h^x\bigg(\prod_{m=1}^{n-1}K_{[m-1,m]}\bigg)^*
w_0^{1/2}\\
&=&\cw_{n-1]}.
\end{eqnarray*}

From the above argument and \eqref{eq1}, one can show that
$\cw_{n]}$ is density operator, i.e. $\tr(\cw_{n]})=1$.

Let us show that the defined state $\ffi_{w_0,{\bh}}$ is a forward
QMC. Indeed, define quasi-conditional expectations $\ce_{\L_n^c}$
as follows:
\begin{eqnarray}\label{E-n1}
&&\hat\ce_{\L_1^c}(x_{[0})=\tr_{[1}(K_{[0,1]}w_0^{1/2}x_{[0}w_0^{1/2}K_{[0,1]}^*), \ \ x_{[0}\in \cb_{\L_0^c}\\
&&\label{E-n2}
\ce_{\L_k^c}(x_{[k-1})=\tr_{[n}(K_{[k-1,k]}x_{[k-1}K_{[k-1,k]}^*),
\ \ x_{[k-1}\in\cb_{\L_{k-1}^c}, \ \ k=1,2,\dots,n+1,
\end{eqnarray}
here $\tr_{[n}=\tr_{\L_n^c}$. Then for any monomial
$a_{\L_1}\otimes a_{W_2}\otimes\cdots\otimes
a_{W_{n}}\otimes\id_{W_{n+1}}$, where
$a_{\L_1}\in\cb_{\L_1},a_{W_k}\in\cb_{W_k}$, ($k=2,\dots,n$), we
have
\begin{eqnarray}\label{E-n2}
\ffi^{(n)}_{w_0,\bh}(a_{\L_1}\otimes a_{W_2}\otimes\cdots\otimes
a_{W_{n}})&=&\tr\bigg(\bh_{n+1}K_{[n,n+1]}^*\cdots
K_{[0,1]}^*w_0^{1/2}(a_{\L_1}\otimes a_{W_2}\otimes\cdots\otimes
a_{W_{n}})\nonumber\\
&&w_0^{1/2}K_{[0,1]}\cdots K_{[n,n+1]}\bigg)\nonumber\\
&=&\tr_{[1}\bigg(\bh_{n+1}K_{[n,n+1]}^*\cdots
K^*_{[1,2]}\hat\ce_{\L_1^c}(a_{\L_1})a_{W_2}K_{[1,2]}\nonumber\\
&&\cdots a_{W_{n}}K_{[n,n+1]}\bigg)\nonumber\\
&=&\tr_{[n+1}\big(\bh_{n+1}\ce_{\L_{n+1}^c}\circ\ce_{\L_n^c}\circ\cdots\nonumber\\
&& \ce_{\L_2^c}\circ\hat\ce_{\L_1^c}(a_{\L_1}\otimes
a_{W_2}\otimes\cdots\otimes a_{W_{n}})\big).
\end{eqnarray}

Hence, for any $a\in\L\subset\L_{n+1}$ from \eqref{ffi-ff} with
\eqref{K11},\eqref{K2}, \eqref{E-n1}-\eqref{E-n2} one can see that
\begin{equation}\label{qqq}
\ffi^{(n)}_{w_0,\bh}(a)=\tr_{[n+1}\big(\bh_{n+1}\ce_{\L_{n+1}^c}\circ\ce_{\L_n^c}\circ\cdots
\ce_{\L_2^c}\circ\hat\ce_{\L_1^c}(a)\big).
\end{equation}
The projectivity of ${\mathcal{W}}_{n]}$ yields the equality
\eqref{eq4.1re} for $\ffi^{(n)}_{w_0,\bh}$, therefore, from
\eqref{qqq} we conclude that $\ffi_{w_0,\bh}$ is a forward QMC.

\end{proof}

\begin{corollary}\label{state^nwithW_n}
If \eqref{eq1},\eqref{eq2} are satisfied then one has
$\ffi^{(n)}_{w_0,\bh}(a)=\tr(\cw_{n]}(a))$ for any $a\in
\cb_{\Lambda_n}$.
\end{corollary}

\begin{remark} If one takes
$L_{>(x,i),(x,i+1)<}=\id$, $M_{(x,(x,i),(x,i+1))}=\id$ for all
$x\in L$, then we get a QMC constructed in \cite{AMSa}. Therefore,
the provided construction extensions ones given in
\cite{AMSa,AMSa2}.
\end{remark}

Our goal in this paper is to establish the existence of phase
transition for the given family of operators. Heuristically, the
phase transition means the existence of two distinct QMC. Let us
provide a more exact definition.

\begin{definition}
We say that there exists a phase transition for a family of
operators $\{K_{<x,(x,i)>}\}$, $\{L_{>(x,i),(x,i+1)<}\}$,
$\{M_{(x,(x,i),(x,i+1))}\}$ if the following conditions are
satisfied:
\begin{enumerate}
\item[(a)] {\sc existence}: The equations \eqref{eq1}, \eqref{eq2}
have at least two $(u_0,\{h^x\}_{x\in L})$ and $(v_0,\{s^x\}_{x\in
L})$ solutions;

\item[(c)] {\sc not overlapping supports}: there is a projector
$P\in B_L$ such that $\ffi_{u_0,\bh}(P)<\ve$ and
$\ffi_{v_0,\bs}(P)>1-\ve$, for some $\ve>0$.

\item[(b)] {\sc non quasi-equivalence}: the corresponding quantum
Markov chains $\ffi_{u_0,\bh}$ and $\ffi_{v_0,\bs}$ are not quasi
equivalent\footnote{Recall that a representation $\pi_1$ of a
$C^*$-algebra $\ga$ is \textit{normal} w.r.t. another
representation $\pi_2$, if there is a normal $*$- epimorphism
$\rho:\pi_2(\ga)''\to \pi_1(\ga)''$ such that
$\rho\circ\pi_2=\pi_1$. Two representations $\pi_1$ and $\pi_2$
are called \textit{quasi-equivalent} if $\pi_1$ is normal w.r.t.
$\pi_2$, and conversely, $\pi_2$ is normal w.r.t. $\pi_1$. This
means that there is an isomorphism  $\alpha:\pi_1(\ga)''\to
\pi_2(\ga)''$ such that $\alpha \circ\pi_1=\pi_2$. Two states
 $\varphi$ and $\psi$ of $\ga$ are said be \textit{quasi-equivalent} if
 the GNS representations $\pi_\varphi$ and $\pi_\psi$ are
 quasi-equivalent \cite{BR}.}.

\end{enumerate}
Otherwise, we say there is no phase transition.
\end{definition}

\begin{remark} We stress that in \cite{AMSa2} we have first introduced a notion of the phase
transition quantum Markov chains. That definition is contained only
(a) and (b) conditions. After some discussions it was observed that
these two conditions are not sufficient for the existence of the
phase transition. Since, the non quasi equivalence of the states
does not imply that the states are different. Namely, it might
happen that one of the states could be absolutely continuous to
another one. In general, to have a phase transition the states
should not be absolutely continuous to each other. Therefore, the
third condition (c) is imposed in the present definition.
\end{remark}

%%%%%%%%%%%%%%%%%%%%%%%%%%%%%%%%%%%%%%%%%%%%%%%%%%%%%%%%%%%%%%%%%%%%%%%%%%%%%%%%%%%%%%%%%%%%%%%%%%%%%%%%%%%%

\section{QMC associated with Ising model with competing interactions}\label{exam1}
%%%%%%%%%%%%%%%%%%%%%%%%%%%%%%%%%%%%%%%%%%%%%%%%%%%%%%%%%%%%%%%%%%%%%%%%%%%%%%%%%%%%%%%%%%%%%%%%%%%%%%%%%%%%%%%%%ù

In this section, we define the model and formulate the main
results of the paper. In what follows we consider a semi-infinite
Cayley tree $\G^2_+=(L,E)$ of order two. Our starting
$C^{*}$-algebra is the same $\cb_L$ but with $\cb_{x}=M_{2}(\bc)$
for all $x\in L$. Denote
\begin{equation}\label{pauli} \id^{(u)}=\left(
          \begin{array}{cc}
            1 & 0 \\
            0 & 1 \\
          \end{array}
        \right), \quad
        \         \
\sigma^{(u)}= \left(
          \begin{array}{cc}
            1 & 0 \\
            0 & -1 \\
          \end{array}
        \right).
\end{equation}

For every vertices  $(x,(x,1),(x,2))$  we put
\begin{eqnarray}\label{1Kxy1}
&&K_{<x,(x,i)>}=\exp\{\b H_{x,(x,i)>}\}, \ \ i=1,2,\ \b>0,\\[2mm] \label{1Lxy1} &&
L_{>(x,1),(x,2)<}=\exp\{J\beta H_{>(x,1),(x,2)<}\}, \ \ J>0,
\end{eqnarray}
where
\begin{eqnarray}\label{1Hxy1}
&&
H_{<x,(x,i)>}=\frac{1}{2}\big(\id^{x)}\id^{(x,i)}+\sigma^{(x)}\sigma^{(x,i)}\big),\\[2mm]
\label{1H>xy<1} &&
H_{>(x,1),(x,2)<}=\frac{1}{2}\big(\id^{(x,1)}\id^{(x,2)}+\sigma^{(x,1)}\sigma^{(x,2)}\big).
\end{eqnarray}

Furthermore, we assume that $M_{(x,(x,i),(x,i+1))}=\id$
($i=1,2,\dots,k$) for all $x\in L$.

The defined model is called  the {\it Ising model with competing
interactions} per vertices  $(x,(x,1),(x,2))$.

\begin{remark}
Note that if we take $J=0$, then one gets the Ising model on Cayley
tree which has been studied in \cite{AMSa3}. In \cite{MBS1} we have
studied the classical Ising model with competing interactions in
comparison with quantum analogous.
\end{remark}

One can calculate that
\begin{eqnarray}\label{1Hxym}
&&H_{<u,v>}^{m}=H_{<u,v>}=\frac{1}{2}\big(\id^{(u)}\id^{(v)}+\sigma^{(u)}\sigma^{(v)}\big),\\[2mm]
\label{1Hxym}
&&H_{>x,y<}^{m}=H_{>x,y<}=\frac{1}{2}\big(\id^{(x)}\id^{(y)}+\sigma^{(x)}\sigma^{(y)}\big).
\end{eqnarray}
Therefore, one finds
\begin{eqnarray}\label{K<u,v>K_0K_3}
&&K_{<u,v>}=K_0\id^{(u)}\id^{(v)}+K_3\sigma^{(u)}\sigma^{(v)},\\[2mm]
&&\label{L>u,v<R_0R_3}
L_{>u,v<}=R_0\id^{(u)}\id^{(v)}+R_3\sigma^{(u)}\sigma^{(v)},
\end{eqnarray}
where
\begin{eqnarray*}
&&K_0=\frac{\exp{\beta}+1}{2},\ \ \   K_3=\frac{\exp{\beta}-1}{2},\\[2mm]
&&R_0=\frac{\exp{(J\beta)}+1}{2}, \ \ \
R_3=\frac{\exp{(J\beta)}-1}{2}.
\end{eqnarray*}
Hence, from \eqref{K1} for each $x\in L$ we obtain
\begin{eqnarray}\label{Ax}
A_{(x,(x,1),(x,2))}&=&\gamma\id^{(x)}\otimes\id^{(x,1)}\otimes\id^{(x,2)}+\delta\sigma^{(x)}\otimes\sigma^{(x,1)}\otimes\id^{(x,2)}\nonumber\\[2mm]
&&+
\delta\sigma^{(x)}\otimes\id^{(x,1)}\otimes\sigma^{(x,2)}+\eta\id^{(x)}\otimes\sigma^{(x,1)}\otimes\sigma^{(x,2)},
\end{eqnarray}
where
\begin{equation}
\left\{
  \begin{array}{ll}
    \gamma=K_{0}^{2}R_{0}+K_{3}^{2}.R_{3}=\frac{1}{4}[\exp{(J+2)\beta}+\exp{J\beta}+2\exp{\beta}], \\
    \\
    \delta=K_{0}K_{3}(R_{0}+R_{3})=\frac{1}{4}\exp{J\beta}[\exp{2\beta}-1],  \\
    \\
    \eta=K_{0}^{2}R_{3}+K_{3}^{2}R_{0}=\frac{1}{4}[\exp{(J+2)\beta}+\exp{J\beta}-2\exp{\beta}]. \\
  \end{array}
\right.
\end{equation}

Recall that a function $\{h^x\}$ is called
\textit{translation-invariant} if one has $h^{x}=h^{\tau_g(x)}$,
for all $x,g\in \G^2_+$. Clearly, this is equivalent to $h^x=h^y$
for all $x,y\in L$.

In what follows, we restrict ourselves to the description of
translation-invariant solutions of \eqref{eq1},\eqref{eq2}.
Therefore, we assume that: $ h^{x}=h$ for all $x\in L$, where
$$
h= \left(
\begin{array}{cc}
            h_{11} & h_{12} \\
            h_{21} & h_{22} \\
          \end{array}
\right).$$

Then we have
\begin{eqnarray*}\label{AHHA}
&& A_{(x,(x,1),(x,2))}\times[\id^{(x)}\otimes h^{(x,1)}\otimes h^{(x,2)}]\times A_{(x,(x,1),(x,2))} \nonumber\\
\\
&=&\big[\gamma^2\id\otimes h\otimes h + \delta^2\id\otimes \sigma h\sigma \otimes h + \delta^2\id\otimes h \otimes \sigma h \sigma +
\eta^2\id\otimes\sigma h \sigma\otimes\sigma h \sigma\big]\nonumber\\
\\
&+&\big[\gamma \eta\id\otimes h \sigma\otimes h \sigma + \gamma
\eta\id\otimes\sigma h\otimes
\sigma h + \delta^2\id\otimes\sigma h \otimes h \sigma + \delta^2\id\otimes h \sigma\otimes\sigma h\big]\nonumber\\
\\
& +&\big[\gamma\delta\sigma\otimes h\sigma \otimes h+\gamma \delta\sigma\otimes h\otimes h\sigma+\gamma\delta\sigma\otimes h\otimes\sigma h+
\gamma\delta\sigma\otimes\sigma h\otimes h\big] \nonumber\\
 \\
&+&\big[\delta\eta\sigma\otimes\sigma h \sigma\otimes h \sigma+\delta\eta\sigma\otimes h\sigma\otimes\sigma h \sigma+\delta\eta\sigma\otimes\sigma h\sigma\otimes\sigma h + \delta\eta\sigma\otimes \sigma h \otimes\sigma h \sigma\big].\nonumber\\
\end{eqnarray*}
Hence, from the last equality we can rewrite \eqref{eq2} as
follows
\begin{eqnarray} \label{eqder}
h&=&Tr_{x]}A_{(x,(x,1),(x,2))}[\id^{(x)}\otimes h\otimes h]A_{(x,(x,1),(x,2))}^{*}\nonumber\\
&=&\tau_1{\Tr( h)}^2+\tau_2{\Tr(\sigma h)}^2\id^{(x)} +
\tau_3\Tr(h)\Tr(\sigma h)\sigma^{(x)}.
\end{eqnarray}
 Here $\theta=\exp2\beta>0$ and
$$\left\{
   \begin{array}{ll}
     \tau_1:=\gamma^2+2\delta^2+\eta^2=\frac{1}{4}[\theta^J(\theta^2+1)+2\theta],\\
      \\
     \tau_2:=2(\gamma\eta+\delta^2)=\frac{1}{4}[\theta^J(\theta^2+1)-2\theta],\\
      \\
     \tau_3:=4\delta(\gamma+\eta)=\frac{1}{2}\theta^J(\theta^2-1),
   \end{array}
 \right.$$
Now taking into account
$$\Tr(h)=\frac{h_{11}+h_{22}}{2},\   \
\Tr(\sigma h)=\frac{h_{11}-h_{22}}{2}$$ the equation \eqref{eqder}
is reduced to the following one
\begin{equation}\label{EQ1}
\left\{
   \begin{array}{lll}
\Tr(h)=\tau_1\Tr(h)^2+\tau_2\Tr(\s h)^2,\\
\Tr(\s h)=\tau_3\Tr(h)\Tr(\s h),\\
h_{21}=0, h_{12}=0.\\
   \end{array}
 \right.
 \end{equation}

The obtained equation implies that a solution $h$ is diagonal, and
$\omega_{0}$ could be also chosen diagonal, through the equation.
In what follows, we always assume that $h_{21}=0, h_{12}=0$. In
the next sections we are going to examine \eqref{EQ1}.

%%%%%%%%%%%%%%%%%%%%%%%%%%%%%%%%%%%%%%%%%%%%%%%%%%%%%%%%%%%%%%%%%%%%%%%%%%%%%%%%%%%%%%%%%%%%%%%%%%%%%%%%%%%%%%%%%%%%%%%%%%%%%%%%%%%%%%%
\section{Existence of QMC associated with the model.}
%%%%%%%%%%%%%%%%%%%%%%%%%%%%%%%%%%%%%%%%%%%%%%%%%%%%%%%%%%%%%%%%%%%%%%%%%%%%%%%%%%%%%%%%%%%%%%%%%%%%%%%%%%%%%%%%%%%%%%%%%%%%%%%%%%%%%%%
In this section we are going to solve \eqref{EQ1}, which yields
the existence of QMC associated with the model.

%%£££££££££££££££££££££££££££££££££££££££££££££££££££££££££££££
\subsection{Case $h_{11}=h_{22}$ and associate QMC}
%%£££££££££££££££££££££££££££££££££££££££££££££££££££££££££££££

Assume that $h_{1,1}=h_{2,2}$ , then \eqref{EQ1} is reduced to
\begin{equation*}
    h_{11}=h_{22}=\frac{1}{\tau_1}.
\end{equation*}

Then putting $\alpha=\frac{1}{\tau_1}$ we get
\begin{equation}
h_{\alpha}=
 \left(
  \begin{array}{cc}
    \alpha & 0 \\
    0 & \alpha \\
  \end{array}
\right)
\end{equation}

\begin{proposition}\label{5.1}
The pair $(\omega_{0},\{h^{x}=h_{\alpha}| x\in L\})$ with
$\omega_{0}=\frac{1}{\alpha}\id,\   \  h^{x}=h_{\alpha}, \forall
x\in L,$ is solution of \eqref{eq1},\eqref{eq2}. Moreover the
associated QMC  can be written on the local algebra
$\mathcal{B}_{L, loc}$ by:
\begin{equation}
\varphi_{\alpha}(a)=\alpha^{2^{n}-1}\Tr\bigg(a\prod_{i=0}^{n-1}K_{[i,i+1]}K_{[i,i+1]}^{*}\bigg),
\ \ \forall a\in B_{\Lambda_{n}}.
\end{equation}
\end{proposition}

\begin{proof}
 Let $n\in\natural, a\in B_{\Lambda_{n}}$
then using Lemma \ref{LL} one finds
\begin{eqnarray*}
\varphi_{\alpha}(a) &=& \Tr\bigg(\omega_{0}h_{n}\bigg(\prod_{m=0}^{n-1}K_{[m,m+1]}K_{[m,m+1]}^{*}\bigg)a\bigg) \nonumber\\
&=&\alpha^{|W_{n}|-1}\Tr\bigg(a \prod_{m=0}^{n-1}K_{[m,m+1]}K_{[m,m+1]}^{*}\bigg) \nonumber\\
&=&\alpha^{2^{n}-1}\Tr\bigg(a
\prod_{i=0}^{n-1}K_{[i,i+1]}K_{[i,i+1]}^{*} \bigg).
\end{eqnarray*}
This completes the proof.
\end{proof}

%%%£££££££££££££££££££££££££££££££££££££££££££££££££££££££££££££££££££
\subsection{Case  $h_{11}\neq h_{22}$ and associate QMC}
%%%£££££££££££££££££££££££££££££££££££££££££££££££££££££££££££££££££££

Assume that $h_{11}\neq h_{22}$, then \eqref{EQ1} is reduced to
\begin{eqnarray}\label{EQ2}
\left\{
  \begin{array}{ll}\
    h_{11}+h_{22}=\frac{1}{\tau_2},  \\
        (h_{11}-h_{22})^{2}=\frac{\tau_{3}-\tau_1}{\tau_2.\tau_{3}^{2}},
  \end{array} \right.
\end{eqnarray}
Let $$\Delta(\theta):=4(\tau_3-\tau_1)=\theta^{J}(\theta^{2}-3)-2\theta$$

One can see that the last system has a solution iff
$\tau_3>\tau_1$, i.e. whenever $\Delta(\theta)>0$.

\begin{proposition}\label{hh0}
Assume that $\tau_3>\tau_1$. Then the equation \eqref{EQ1} has two
solutions given by:
\begin{eqnarray}\label{hh1}
  && h =\xi_{0}\id+\xi_{3}\sigma,\\
&&\label{hh2}
   h'=\xi_{0}\id-\xi_{3}\sigma,
\end{eqnarray}
where
\begin{eqnarray}\label{xi01}
\xi_{0}=\frac{1}{\tau_3}=\frac{2}{\theta^{J}(\theta^{2}-1)}, \ \ \
\xi_{3}=\frac{\sqrt{\tau_3-\tau_1}}{\tau_3\sqrt{\tau_2}}=\frac{2}{\theta^{J}(\theta^{2}-1)}\sqrt\frac{{\Delta(\theta)}}{{\theta^{J}(\theta^{2}+1)-2\theta}}
\end{eqnarray}
\end{proposition}

\begin{proof}
Assume that  $\Delta(\theta)>0$. Then one can conclude that
\eqref{EQ2} is equivalent to the following system
\begin{eqnarray*}
\left\{
  \begin{array}{ll}
    h_{1,1}+h_{2,2}=2\xi_0,  \\
      h_{1,1}-h_{2,2}=\pm2\xi_3
  \end{array}
\right.
\end{eqnarray*}
It is easy to see that $h_{1,1}=\xi_{0}-\xi_{3}$,
$h_{2,2}=\xi_{0}+\xi_3$. Hence, we get \eqref{hh1},\eqref{hh2}.
\end{proof}

From \eqref{eq1} we find that $\omega_0=\frac{1}{\xi_0}\id\in
\mathcal{B}^{+}$s. Therefore, the pairs $\big(\omega_{0},\ \
\{h^{(x)}=h, \ x\in L\}\big)$ and $\big(\omega_0,\{h^{(x)}=h', \
x\in L\}\big)$ define two solutions of \eqref{eq1},\eqref{eq2}.
Hence, they define two QMC $\varphi_1$ and $\varphi_2$,
respectively. Namely, for every $ a\in\mathcal{B}_{\Lambda_n}$ one
has
\begin{eqnarray}\label{F1}
&&\varphi_1(a)=\Tr \big(\omega_0 K_{[0,1]}\cdots K_{[n-1,n]} \bh_n
K_{[n-1,n]}^{*}\cdots
K_{[0,1]}^{*}a\big)\\[2mm] \label{F2}
&&\varphi_2(a)=\Tr \big(\omega_0 K_{[0,1]}\cdots K_{[n-1,n]}
\bh'_n K_{[n-1,n]}^{*}\cdots K_{[0,1]}^{*}a\big).
\end{eqnarray}

Hence, we have proved the following

\begin{theorem}\label{5.3}
Let $\theta>1$. Then the following statements hold:
\begin{itemize}
\item[(i)] if $\Delta(\theta)\leq0$, then there is a unique
translation invariant QMC $\varphi_\a$;

\item[(ii)] if $\Delta(\theta)>0$, then there are at least three
translati0n invariant QMC $\varphi_\a$, $\varphi_1$ and
$\varphi_2$.
\end{itemize}
\end{theorem}

We point out that the critical line $\Delta(\theta)=0$ is first
observed in \cite{GPW} for the classical Ising model with
competing interactions.

Now let us consider some particular values of $\theta$.
\begin{description}
 \item[(i)] Let $J=1$, then one has $\Delta(\theta)=\theta^{3}-5\theta$.
 So,
  \begin{itemize}
   \item if $1<\theta\leq \sqrt{5}$, there is a unique QMC;
   \item if $\theta> \sqrt{5}$, there exist three QMC.
   \end{itemize}
  \item[(ii)] Let $J=2$, then
  $\Delta(\theta)=\theta(\theta+1)(\theta-\frac{1-\sqrt{5}}{2})(\theta-\frac{1+\sqrt{5}}{2})$.
  Hence,
      \begin{itemize}
   \item if $1<\theta\leq \frac{1+\sqrt{5}}{2}$, there is a unique QMC;
   \item if $\theta> \frac{1+\sqrt{5}}{2}$, there exist three QMC.
   \end{itemize}
\end{description}

%%%%%%%%%%%%%%%%%%%%%%%%%%%%%%%%%%%%%%%%%%%%%%%%%%%%%%%%%%%%%%%%%%%%%%%%%%%%%%%%%%%%%%%%%%%%%%%%%%%%%%%%%%%%
%%%%%%%%%%%%%%%%%%%%%%%%%%%%%%%%%%%%%%%%%%%%%%%%%%%%%%%%%%%%%%%%%%%%%%%%%%%%%%%%%%%%%%%%%%%%%%%%%%%%%%%%%%
\section{Proof of Theorem \ref{Main}}
%%%%%%%%%%%%%%%%%%%%%%%%%%%%%%%%%%%%%%%%%%%%%%%%%%%%%%%%%%%%%%%%%%%%%%%%%%%%%%%%%%%%%%%%%%%%%%%%%%%%%%%%%%%%
%%%%%%%%%%%%%%%%%%%%%%%%%%%%%%%%%%%%%%%%%%%%%%%%%%%%%%%%%%%%%%%%%%%%%%%%%%%%%%%%%%%%%%%%%%%%%%%%%%%%%%%%%%%%%ù

This section is devoted to the proof of Theorem \ref{Main}. To
realize it, we first show not overlapping supports of the states
$\varphi_1$ and $\varphi_2$. Then we will show that they are not
quasi-equivalent.

\subsection{Not overlapping supports of $\varphi_1$ and $\varphi_2$}

As usual we put
$$
e_{11}= \left(
\begin{array}{cc}
                  1 & 0 \\
                  0 & 0
                \end{array}
                \right),  \    \           \
                e_{22}=
                \left( \begin{array}{cc}
                  0 & 0 \\
                  0 & 1
                \end{array}
                \right).
$$

Now for each $n\in\mathbf{N}$ denote
$$
p_n:=\bigg(\bigotimes_{x\in \Lambda_n}e_{11}^{(x)}\bigg)\otimes\id, \             \ q_n:=\bigg(\bigotimes_{x\in \Lambda_n}e_{22}^{(x)}\bigg)\otimes\id.
$$

\begin{lemma}\label{pq_n} For every $n\in\bn$ one has
\begin{enumerate}
\item[(i)] $\varphi_{1}(p_n)=\varphi_2(q_n)=\frac{1}{\xi_0}\left(\xi_0+\xi_3\right)^{2^{n}}\left(\frac{\tau_1+\tau_2+\tau_3}{4}\right)^{2^{n}-1},$
\item[(ii)] $\varphi_1(q_n)=\varphi_2(p_n)=\frac{1}{\xi_0}\left(\xi_0-\xi_3\right)^{2^{n}}\left(\frac{\tau_1+\tau_2+\tau_3}{4}\right)^{2^{n}-1}.$
\end{enumerate}
\end{lemma}

\begin{proof} (i). From \eqref{F1} we find
\begin{eqnarray*}
\varphi_1(p_n)=\Tr\left[ \omega_0K_{[0,1]}\cdots
K_{[n-2,n-1]}\cdot\Tr_{n-1]}
\big((K_{[n-1,n]}\bh_np_nK_{[n-1,n]}^{*}\big)\cdot
K_{[n-2,n-1]}^{*}\cdots
K_{[0,1]}^{*}\right]\nonumber\\
\end{eqnarray*}

Now using the fact that: $he_{11}=(\xi_0+\xi_3)e_{11}$ and
\eqref{Ax} one gets
\begin{eqnarray*} \Tr_{n-1]}
K_{[n-1,n]}\bh_np_nK_{[n-1,n]}^{*}
&=& p_{n-2}\otimes\prod_{x\in  \overrightarrow W_{n-1}}A_{(x,(x,1),(x,2))}e_{11}^{(x)}\otimes h e_{11}^{(x,1)}\otimes h e_{11}^{(x,2)}A_{(x,(x,1),(x,2))}\nonumber\\
&=& (\xi_0+\xi_3)^{2|W_{n-1]}|}\bigg(\frac{\tau_1+\tau_2+\tau_3}{4}\bigg)^{|W_{n-1}|}p_{n-1}.\nonumber\\
\end{eqnarray*}
Hence,
\begin{eqnarray*}
\varphi_1(p_n)&=&(\xi_0+\xi_3)^{|W_{n}|}\left(\frac{\tau_1+\tau_2+\tau_3}{4}\right)^{|W_{n-1}|}\Tr\left[ \omega_0K_{[0,1]}\cdots K_{[n-2,n-1]}p_{n-1}K_{[n-2,n-1]}^{*}\cdots K_{[0,1]}^{*}\right]\nonumber\\
&&.\nonumber\\
&&.\nonumber\\
&&.\nonumber\\
&=& \left(\xi_0+\xi_3\right)^{|W_{n}|}\left(\frac{\tau_1+\tau_2+\tau_3}{4}\right)^{|W_{n-1}|+ ... +|W_{0}|}.\Tr\left[ \omega_0p_{0}\right]\nonumber\\
&=&\left(\xi_0+\xi_3\right)^{|W_{n}|}\left(\frac{\tau_1+\tau_2+\tau_3}{4}\right)^{|\Lambda_{n-1}|}\Tr\left[ \omega_0p_{0}\right]\nonumber\\
&=&\frac{1}{\xi_0}\left(\xi_0+\xi_3\right)^{2^{n}}\left(\frac{\tau_1+\tau_2+\tau_3}{4}\right)^{2^n-1}\left(\frac{1}{\xi_0}+\frac{1}{\xi_3}\right).
\end{eqnarray*}

Analogously, using the fact that: $h'e_{22}=(\xi_0+\xi_3)e_{22}$
we obtain

\begin{eqnarray*}
\Tr_{n-1]} K_{[n-1,n]}\bh'_nq_nK_{[n-1,n]}^{*}
&=& q_{n-2}\otimes\prod_{x\in \mathcal{W}_{n-1]}}A_{(x,(x,1),(x,2))}e_{2,2}^{(x)}\otimes h' e_{2,2}^{(x,1)}\otimes h' e_{22}^{(x,2)}A_{(x,(x,1),(x,2))}\nonumber\\
&=&
(\xi_0+\xi_3)^{2|\mathcal{W}_{n-1]}|}\bigg(\frac{\tau_1+\tau_2+\tau_3}{4}\bigg)^{|\mathcal{W}_{n-1]}|}q_{n-1}.
\end{eqnarray*}
which yields
$$\varphi_2(q_n)=\frac{1}{\xi_0}\left(\xi_0+\xi_3\right)^{2^{n}}\left(\frac{\tau_1+\tau_2+\tau_3}{4}\right)^{2^n-1}\left(\frac{1}{\xi_0}+\frac{1}{\xi_3}\right).$$

(ii) Now from $he_{22}=(\xi_0-\xi_3)e_{22}$ and
$h'e_{11}=(\xi_0-\xi_3)e_{11}$. One can find.
\begin{eqnarray*}
&&\Tr_{n-1]} K_{[n-1,n]}\bh_nq_nK_{[n-1,n]}^{*} =
(\xi_0-\xi_3)^{2|W_{n-1}|}\bigg(\frac{\tau_1+\tau_2+\tau_3}{4}\bigg)^{|W_{n-1}|}q_{n-1}.\\[2mm]
&&\Tr_{n-1]} K_{[n-1,n]}\bh'_np_nK_{[n-1,n]}^{*} =
(\xi_0-\xi_3)^{2|W_{n-1}|}\bigg(\frac{\tau_1+\tau_2+\tau_3}{4}\bigg)^{|W_{n-1}|}p_{n-1}.
\end{eqnarray*}
The same argument as above implies (ii). This completes the proof.
\end{proof}

\begin{theorem}\label{6.2}
For $n\in\mathbf{N}$ fixed, one has
$$\varphi_1(p_n) \rightarrow 1, \      \   \varphi_2(q_n)\rightarrow 1 \ \ \textrm{as} \ \ \beta\rightarrow +\infty.
$$
\end{theorem}

\begin{proof} We know that $\theta=\exp (2\beta)\rightarrow+\infty $ as
$\beta\rightarrow+\infty$. Hence, one finds
\begin{eqnarray*}
&& \frac{1}{\xi_0}=\frac{\theta^J(\theta^2-1)}{2}\sim
\frac{\theta^{J+2}}{2}, \ \ \textrm{as} \ \
\theta\to+\infty\\[2mm]
&&
\left(\xi_0+\xi_3\right)^{2^{n}}=\left(\frac{2}{\theta^J(\theta^2-1)}(1+\sqrt{\frac{\theta^J(\theta^2-3)-2\theta}{\theta^J(\theta^2+1)-2\theta}})\right)^{2^{n}}
\sim \left(\frac{4}{\theta^{J+2}}\right)^{2^{n}}, \ \ \textrm{as}
\ \
\theta\to+\infty\\[2mm]
&&\left(\frac{\tau_1+\tau_2+\tau_3}{4}\right)^{2^{n}-1}=\left(\frac{\theta^{J+2}}{4}\right)^{2^{n}-1}.
\end{eqnarray*}

Hence, we obtain
$$\varphi_{1}(p_n)=\varphi_2(q_n)
\sim
\frac{\theta^{J+2}}{4}\left(\frac{4}{\theta^{J+2}}\right)^{2^{n}}\left(\frac{\theta^{J+2}}{4}\right)^{2^{n}-1}=1.
$$
So, $$ \lim_{\theta\to\infty}\varphi_{1}(p_n)=
\lim_{\theta\to\infty} \varphi_2(q_n)=1.$$ This completes the
proof.
\end{proof}

\begin{remark} We note that from  $p_n\leq 1-q_n$ one gets
$$
\lim_{\theta\to\infty}\varphi_{1}(q_n)= \lim_{\theta\to\infty}
\varphi_2(p_n)=0.$$

This implies that the states $\varphi_1$ and $\varphi_2$ have not
overlapping supports.
\end{remark}

\subsection{Non quasi equivalence of $\varphi_1$ and
$\varphi_2$}

In this subsection we are going to proof that the state
$\varphi_1$ and $\varphi_2$ are not quasi equivalent.

First note that from the construction of the states $\varphi_1$ and
$\varphi_2$ one can see that they are translation invariant, i.e.
$\varphi_i\tau_g=\varphi_i$ ($i=1,2$) (see \eqref{trans1}) for all
$g\in \G^2_+$. Moreover, using the same argument as in \cite{[Sp75]}
one can show that states $\varphi_1$ and $\varphi_2$ satisfy mixing
property, i.e.
$$
\lim_{|g|\to\infty}\varphi_i(\tau_g(x)y)=\varphi_i(x)\varphi_i(y),
\ \ i=1,2.
$$
This means that they are factor states.

 To establish the
non-quasi equivalence, we are going to use the following result
(see \cite[Corollary 2.6.11]{BR}).

\begin{theorem}\label{br-q}
Let $\varphi_1,$ $\varphi_2$ be two factor states on a quasi-local
algebra $\ga=\cup_{\Lambda}\ga_\Lambda$. The states $\varphi_1,$
$\varphi_2$ are  quasi-equivalent if and only if for any given
$\varepsilon>0$ there exists a finite volume $\Lambda\subset L$
such that $\|\varphi_1(a)-\varphi_2(a)\|<\varepsilon \|a\|$ for
all $a\in B_{\Lambda^{'}}$ with
$\Lambda^{'}\cap\Lambda=\emptyset.$
\end{theorem}

Let us define an element of $\mathcal{B}_{\Lambda_n}$ as follows:
$$
E_{\Lambda_n}:=e_{11}^{x_{W_n}^{(1)}}\otimes\bigg(\bigotimes_{y\in
\Lambda_n\setminus \{x_{W_n}^{(1)}\}}\id^{y}\bigg),
$$
where $x_{W_n}^{(1)}$ is defined in \eqref{xw}.

Now we are going to calculate $\varphi_1(E_{\Lambda_n})$ and
$\varphi_2(E_{\Lambda_n})$, respectively.

First consider the state $\varphi_1$, then we know that this state
is defined by  $\omega_0=\tau_3\id$ and
$h^{x}=h=\xi_0\id+\xi_3\sigma$. Define two  elements of
$\mathcal{B}_{W_n}$ by
$$\hat{\bh}_n:=\id^{x_{W_n}^{(1)}}\otimes\bigotimes_{x\in W_n\setminus \{x_{W_n}^{(1)}\}}h^{(x)}$$
$$\check{\bh}_{n}:=\sigma^{x_{W_n}^{(1)}}\otimes\bigotimes_{x\in W_n\setminus \{x_{W_n}^{(1)}\}}h^{(x)}$$

\begin{lemma}\label{f-p-11}
Let
$$\hat{\psi}_n:=\Tr_{n-1]}\big[\omega_{0}K_{[0,1]}\cdots K_{[n-1,n]}\hat{\bh}_{n}K_{[n-1,n]}^{*}\cdots K_{[0,1]}^{*}\big]$$
$$\check{\psi}_n:=\Tr_{n-1]}\big[\omega_{0}K_{[0,1]}\cdots K_{[n-1,n]}\check{\bh}_{n}K_{[n-1,n]}^{*}\cdots K_{[0,1]}^{*}\big]$$
Then there are two pairs of reals
$(\hat{\rho}_{1},\hat{\rho}_{2})$ and
$(\check{\rho}_{1},\check{\rho}_{2})$ depending on $\theta$ such
that
$$\left\{
  \begin{array}{ll}
    \hat{\psi}_n=\hat{\rho}_{1}+\hat{\rho}_{2}(\frac{\tau_1}{\tau_3}-1)^{n}, \\
\\
    \check{\psi}_n=\check{\rho}_{1}+\check{\rho}_{2}(\frac{\tau_1}{\tau_3}-1)^{n} \\
  \end{array}
\right.$$
\end{lemma}

\begin{proof} One can see that
\begin{eqnarray*}
\left(
  \begin{array}{ll}
    \hat{\psi}_n \\
\\
    \check{\psi}_n
  \end{array}
\right) &=& \left(
     \begin{array}{c}
       \Tr_{n]}\big[\omega_{0}K_{[0,1]}\cdots K_{[n-2,n-1]}\Tr_{n-1]}[K_{[n-1,n]}\hat{\bh}_{n}K_{[n-1,n]}^{*}]K_{[n-2,n-1]}^{*}\cdots K_{[0,1]}^{*}\big],  \\
       \\
       \Tr_{n]}\big[\omega_{0}K_{[0,1]}\cdots K_{[n-2,n-1]}Tr_{n-1]}[K_{[n-1,n]}\check{\bh}_{n}K_{[n-1,n]}^{*}]K_{[n-2,n-1]}^{*}\cdots K_{[0,1]}^{*}\big]
     \end{array}
\right).
\end{eqnarray*}
After small calculations, we find
 $$\left\{
  \begin{array}{ll}
  \Tr_{x]}\left[A_{(x,(x,1),(x,2))}\big(\id^{(x)}\otimes\id^{(x,1)}\otimes h^{(x,2)}\big)A_{(x,(x,1),(x,2))}\right]=\tau_{1}\xi_0\id^{(x)}+\frac{1}{2}\tau_{3}\xi_3\sigma^{(x)} \\
   \\
   \Tr_{x]}\big[A_{(x,(x,1),(x,2))}\big(\id^{(x)}\otimes\sigma^{(x,1)}\otimes h_{(\xi_{0},\xi_{3})}^{(x,2)}\big)A_{(x,(x,1),(x,2))}\big]=\tau_2\xi_3\id^{(x)}+\frac{1}{2}\sigma^{(x)}
  \end{array}
\right.$$

Hence,  one gets
$$\left\{
\begin{array}{ll}
\Tr_{n-1]}K_{[n-1,n]}\hat{\bh}_{n}K_{[n-1,n]}^{*}=\tau_1\xi_0\hat{h}_{n-1}+\frac{1}{2}\tau_3\xi_3\check{h}_{n-1},\\
\\
\Tr_{n-1]}K_{[n-1,n]}\check{\bh}_{n}K_{[n-1,n]}^{*}=\tau_2\xi_3\hat{h}_{n-1}+\frac{1}{2}\check{h}_{n-1}.
\end{array}
\right.$$

Therefore,
\begin{eqnarray*}
\left(
  \begin{array}{ll}
    \hat{\psi}_n \\[2mm]
    \check{\psi}_n
  \end{array}
\right)
&=&
\left(
\begin{array}{c}
\tau_1\xi_0\hat{\psi}_{n-1}+\frac{1}{2}\tau_3\xi_3\check{\psi}_{n-1}
\\[2mm]
\tau_2\xi_3\hat{\psi}_{n-1}+\frac{1}{2}\check{\psi}_{n-1}
\end{array}
\right)\nonumber\\[2mm]
&=&
\left(
   \begin{array}{cc}
   \tau_1\xi_0 & \        \ \frac{1}{2}\tau_3\xi_3 \\[2mm]
     \tau_2\xi_3 &\        \ \frac{1}{2} \\

     \end{array}
     \right)
\left(
\begin{array}{c}
  \hat{\psi}_{n-1} \\[2mm]

  \check{\psi}_{n-1} \\
\end{array}
\right)\nonumber\\
\vdots\nonumber\\
&=&\left(
   \begin{array}{cc}
   \tau_1\xi_0 & \        \ \frac{1}{2}\tau_3\xi_3 \\[2mm]
     \tau_2\xi_3 &\        \ \frac{1}{2} \\
     \end{array}
     \right)^{n}
\left(
\begin{array}{c}
  \hat{\psi}_{0} \\[2mm]
  \check{\psi}_{0} \\
\end{array}
\right),
\end{eqnarray*}
where $$\left\{
\begin{array}{c}
  \hat{\psi}_{0}=\Tr (\omega_{0})=\frac{1}{\xi_0}\\[2mm]
  \check{\psi}_{0}=\Tr (\omega_{0}.\sigma) = 0 \\
\end{array}
\right.
$$
The matrix
$$N := \left(
   \begin{array}{cc}
   \tau_1\xi_0 & \        \ \frac{1}{2}\tau_3\xi_3 \\[2mm]
     \tau_2\xi_3 &\        \ \frac{1}{2} \\
     \end{array}
     \right)
     $$
     can be written in diagonal form by:

     $$N=P
     \left(
   \begin{array}{cc}
   1  & \        \ 0 \\
     0 &\        \ \frac{\tau_{1}}{\tau3}-\frac{1}{2} \\
     \end{array}
     \right)
     P^{-1}
     $$
     where
     $$
       P=\left(
   \begin{array}{cc}
   \frac{\tau_3}{2\tau_2} & \        \  -\frac{\xi_3}{\xi_0} \\
     \frac{\xi_3}{\xi_0} &\        \ 1 \\
     \end{array}
     \right), \     \  \det(P)=\frac{3\tau_{3}-2\tau_{1}}{2\tau_2}$$

 So,
     \begin{eqnarray*}
     \left(
       \begin{array}{c}
         \hat{\psi}_n \\
               \check{\psi}_n \\
       \end{array}
     \right)
     &=&P \left(
   \begin{array}{cc}
   1  & \        \ 0 \\
     0 &\        \ (\frac{\tau_{1}}{\tau3}-\frac{1}{2} )^{n}\\
     \end{array}
     \right)P^{-1}
     \left(
     \begin{array}{c}
       \frac{1}{\xi_0} \\
       \\
       0
      \end{array}
     \right)\nonumber\\
      &=&
     \left(
  \begin{array}{ll}
    \hat{\rho}_{1}+\hat{\rho}_{2}(\frac{\tau_1}{\tau_3}-\frac{1}{2})^{n}\\[2mm]
    \check{\rho}_{1}+\check{\rho}_{2}(\frac{\tau_1}{\tau_3}-\frac{1}{2})^{n} \\
  \end{array}
\right).
\end{eqnarray*}
where
\begin{eqnarray}\label{r1}
&& \hat{\rho}_{1}=\frac{2\tau_{3}^{2}}{3\tau_{3}-2\tau_{1}},  \ \
\hat{\rho}_{2}=\frac{2\tau_{3}(\tau_3-\tau_1)}{3\tau_{3}-2\tau_{1}},\\[2mm]
&& \check{\rho}_{1}=\frac{2\tau_2 \tau_{3}^{2}\xi_{3}}{
3\tau_3-2\tau_2},\      \   \check{\rho}_{2}=-\frac{2\tau_2
\tau_{3}^{2}\xi_{3}}{ 3\tau_3-2\tau_2}.\label{r2}
\end{eqnarray}
This completes the proof.
\end{proof}

\begin{proposition}\label{6.6} For each $n\in\bn$ one has
\begin{eqnarray*}
\varphi_1(E_{\Lambda_n})&=&\frac{1}{2}\left[(\xi_0+\xi_3)(\xi_0\tau_2+\xi_3\tau_2)\hat{\rho}_{1}+\frac{\tau_3}{2}(\xi_0+\xi_3)^{2}\check{\rho}_{1}\right]\nonumber\\
&&+\frac{1}{2}\left[(\xi_0+\xi_3)(\xi_0\tau_1+\xi_3\tau_2)\hat{\rho}_{2}+\frac{1}{2}\tau_3(\xi_0+\xi_3)^{2}\check{\rho}_{2}\right]\bigg(\frac{\tau_1}{\tau_3}-\frac{1}{2}\bigg)^{n-1}.
\end{eqnarray*}
\end{proposition}

\begin{proof} From \eqref{F1} we have
\begin{eqnarray*}
\varphi_1(E_{\Lambda_n})&=&\Tr\left[ \omega_0K_{[0,1]}\cdots
K_{[n-2,n-1]}\Tr_{n-1]}(K_{[n-1,n]}\bh_n E_{\Lambda_n}
K_{[n-1,n]}^{*})\cdots K_{[0,1]}^{*}\right].
\end{eqnarray*}
One can calculate that
\begin{eqnarray*}
\Tr_{n-1]} (K_{[n-1,n]}\bh_n E_{\Lambda_n}
K_{[n-1,n]}^{*})&=&Tr_{x_{W_{n-1}}^{(1)}]}\bigg(A_{(x_{W_{n-1}}^{(1)},
x_{W_n}^{(1)}, x_{W_n}^{(2)})}
\big(\id^{x_{W_{n-1}}^{(1)}}\otimes e_{1,1}h^{x_{W_n}^{(1)}}\otimes h^{x_{W_n}^{(2)}}\big)\nonumber\\
&&A_{(x_{W_{n-1}}^{(1)}, x_{W_n}^{(1)}, x_{W_n}^{(2)})}^{*}\bigg)\otimes\bigotimes_{x\in W_{n-1}\setminus \{x_{W_n}^{(1)}\}}h^{(x)}\nonumber\\
&=&\frac{1}{2}\left[(\xi_0+\xi_3)(\tau_1\xi_0+\tau_2\xi_3)\hat{\bh}_{n-1}+\frac{\tau_3}{2}(\xi_0+\xi_3)^{2}\check{\bh}_{n-1}\right].
\end{eqnarray*}
Hence,
\begin{eqnarray*}
\varphi_1(E_{\Lambda_n})
&=&\frac{1}{2}(\xi_0+\xi_3)(\tau_1\xi_0+\tau_2\xi_3)\Tr\left[ \omega_0K_{[0,1]}\cdots K_{[n-2,n-1]}\hat{\bh}_{n-1} K_{[n-2,n-1]}^{*}\cdots K_{[0,1]}^{*}\right]\nonumber\\
&&+\tau_3\bigg(\frac{\xi_0+\xi_3}{2}\bigg)^{2}\Tr\left[
\omega_0K_{[0,1]}\cdots K_{[n-2,n-1]}\check{\bh}_{n-1}
K_{[n-2,n-1]}^{*}\cdots
 K_{[0,1]}^{*}\right].\nonumber\\
&=&\frac{1}{2}\left[(\xi_0+\xi_3)(\tau_1\xi_0+\tau_2\xi_3)\hat{\psi}_{n-1}+\frac{\tau_3}{2}(\xi_0+\xi_3)^{2}\check{\psi}_{n-1}\right].%\nonumber\\
\end{eqnarray*}
Now using the values of $\hat{\psi}_{n-1}$  and
$\check{\psi}_{n-1}$ given by the previous lemma we obtain the
result.
\end{proof}

 Now we consider the state $\varphi_2$. Recall that this state is defined by  $\omega_0=\frac{1}{\xi_0}\id$ and
$h^{x}=h'=\xi_0\id-\xi_3\sigma$. Define two elements of
$\mathcal{B}_{W_n}$ by
$$\hat{\bh'}_n:=\id^{x_{W_n}^{(1)}}\otimes\bigotimes_{x\in W_n\setminus \{x_{W_n}^{(1)}\}}h'^{(x)}$$
$$\check{\bh'}_{n}:=\sigma^{x_{W_n}^{(1)}}\otimes\bigotimes_{x\in W_n\setminus \{x_{W_n}^{(1)}\}}h'^{(x)}$$

Using the same argument like in the proof of Lemma \ref{f-p-11} we
can prove the following auxiliary fact.

\begin{lemma}\label{6.7} Let
$$\hat{\phi}_n:=\Tr_{n-1]}\big[\omega_{0}.K_{[0,1]}\cdots K_{[n-1,n]}\hat{\bh'}_{n}K_{[n-1,n]}^{*}\cdots K_{[0,1]}^{*}\big]$$
$$\check{\phi}_n:=\Tr_{n-1]}\big[\omega_{0} K_{[0,1]}\cdots K_{[n-1,n]}\check{\bh'}_{n}K_{[n-1,n]}^{*}\cdots K_{[0,1]}^{*}\big]$$
Then there are two pairs of reals  $(\hat{\pi}_{1},\hat{\pi}_{2})$
and $(\check{\pi}_{1},\check{\pi}_{2})$ depending on $\theta$ such
that
$$\left\{
  \begin{array}{ll}
    \hat{\phi}_n=\hat{\pi}_{1}+\hat{\pi}_{2}(\frac{\tau_1}{\tau_3}-\frac{1}{2})^{n}, \\
    \check{\phi}_n=\check{\pi}_{1}+\check{\pi}_{2}(\frac{\tau_1}{\tau_3}-\frac{1}{2})^{n} \\
  \end{array}
\right.$$ where
$$\hat{\pi}_{1}=\frac{\tau_3^2}{3\tau_3-2\tau_1},  \      \   \hat{\pi}_{2}=\frac{2\tau_{3}(\tau_3-\tau_1)}{3\tau_{3}-2\tau_2},  $$
$$ \check{\pi}_{1}=-\frac{2\tau_{2}\tau_3^2\xi_{3}}{ 3\tau_3-2\tau_1},\      \   \check{\pi}_{2}=\frac{2\tau_{2}\tau_3^{2}\xi_{3}}{ 3\tau_3-2\tau_1}.$$
\end{lemma}

\begin{proposition}\label{6.8} For each $n\in\bn$ one has
\begin{eqnarray*}
\varphi_2(E_{\Lambda_n})&=&\frac{1}{2}\left[(\xi_0-\xi_3)(\xi_0\tau_1-\xi_3\tau_2)\hat{\pi}_{1}+\frac{\tau_3}{2}(\xi_0-\xi_3)^{2}\check{\pi}_{1}\right]\\[2mm]
&&+\frac{1}{2}\left[(\xi_0-\xi_3)(\xi_0\tau_1-\xi_3\tau_2)\hat{\pi}_{2}+\frac{\tau_3}{2}(\xi_0-\xi_3)^{2}\check{\pi}_{2}\right]
\bigg(\frac{\tau_1}{\tau_3}-\frac{1}{2}\bigg)^{n-1}.
\end{eqnarray*}
\end{proposition}

\begin{proof} From \eqref{F2} we find
\begin{eqnarray}\label{6.81}
\varphi_2(E_{\Lambda_n})=\Tr\left[\omega_0K_{[0,1]}\cdots
K_{[n-2,n-1]}\Tr_{n-1]}(K_{[n-1,n]}\bh'_n E_{\Lambda_n}
K_{[n-1,n]}^{*})\cdots K_{[0,1]}^{*}\right].
\end{eqnarray}

We easily calculate that
\begin{eqnarray*}
\Tr_{n-1]} (K_{[n-1,n]}h'_n a_{\Lambda_n}
K_{[n-1,n]}^{*})=\frac{1}{2}(\xi_0-\xi_3)(\tau_1\xi_0-\tau_2\xi_3)\hat{\bh'}_{n-1}+\tau_3\bigg(\frac{\xi_0-\xi_3}{2}\bigg)^{2}\check{\bh'}_{n-1}.
\end{eqnarray*}
Hence, from  \eqref{6.81} one gets
\begin{eqnarray*}
\varphi_2(E_{\Lambda_n})
=\frac{1}{2}\left[(\xi_0-\xi_3)(\tau_1\xi_0-\tau_2\xi_3)\hat{\phi}_{n-1}+\frac{\tau_3}{2}(\xi_0-\xi_3)^{2}\check{\phi}_{n-1}\right].%\nonumber\\
\end{eqnarray*}
Using the values of $\hat{\phi}_{n-1}$  and $\check{\phi}_{n-1}$
given in Lemma \ref{6.7}, we obtain the desired assertion.
\end{proof}

\begin{theorem}\label{6.9}
The two QMC $\varphi_1$ and $\varphi_2$ are not quasi-equivalent.
\end{theorem}

\begin{proof}
For any  $\forall n\in\natural$ it is clear that
$E_{\Lambda_n}\in\mathcal{B}_{\Lambda_n}\setminus\mathcal{B}_{\Lambda_{n-1}}.$
Therefore, for any finite subset $\Lambda\in L$, there exists
$n_0\in\natural$ such that $\Lambda\subset\Lambda_{n_0}$.  Then
for all $n>n_0$ one has
$E_{\Lambda_n}\in\mathcal{B}_{\Lambda_n}\setminus\mathcal{B}_{\Lambda}.$
It is clear that
$$\|E_{\Lambda_n}\|=\|e_{1,1}^{x_{W_n}^{(1)}}\bigotimes_{y\in
L\setminus \{x_{W_n}^{(1)}\}}\id^{y}\|=\|e_{1,1}\|=\frac{1}{2}.$$
From Propositions \ref{6.6} and \ref{6.8} we obtain
\begin{eqnarray*}
\left|\varphi_1(E_{\Lambda_n})-\varphi_1(E_{\Lambda_n})\right|
&=&\frac{1}{2}\bigg|\big[(\xi_0+\xi_3)(\xi_0\tau_1+\xi_3\tau_2)\hat{\rho}_{1}
+\tau_3(\xi_0+\xi_3)^{2}\check{\rho}_{1}\big]\\[2mm]
&&-\big[(\xi_0-\xi_3)(\xi_0\tau_1-\xi_3\tau_2)\hat{\pi}_{1}
+\tau_3(\xi_0-\xi_3)^{2}\check{\pi}_{1}\big]\nonumber\\
&&+\bigg(\big[(\xi_0+\xi_3)(\xi_0\tau_2+\xi_3\tau_2)\hat{\rho}_{2}+\tau_3(\xi_0+\xi_3)^{2}\check{\rho}_{2}\big]\\[2mm]
&&-\big[(\xi_0-\xi_3)(\xi_0\tau_2-\xi_3\tau_2)\hat{\pi}_{2}+\tau_3(\xi_0-\xi_3)^{2}\check{\pi}_{2}\big]\bigg)
\left(\frac{\tau_1}{\tau_3}-\frac{1}{2}\right)^{n-1}\bigg|\nonumber\\
&\geq& I_1-I_2\left|\frac{\tau_1}{\tau_3}-\frac{1}{2}\right|^{n-1}
\end{eqnarray*}
where
\begin{eqnarray*}
I_1&=&\frac{1}{2}\bigg|\big[(\xi_0+\xi_3)(\xi_0\tau_1+\xi_3\tau_2)\hat{\rho}_{1}+\tau_3(\xi_0+\xi_3)^{2}\check{\rho}_{1}\big]\\[2mm]
&&-\big[(\xi_0-\xi_3)(\xi_0\tau_1-\xi_3\tau_2)\hat{\pi}_{1}+\tau_3(\xi_0-\xi_3)^{2}\check{\pi}_{1}\big]\bigg|\\[2mm]
I_2&=&\frac{1}{2}\bigg|\big[(\xi_0+\xi_3)(\xi_0\tau_2+\xi_3\tau_2)\hat{\rho}_{2}+\tau_3(\xi_0+\xi_3)^{2}\check{\rho}_{2}\big]\\[2mm]
&&-\big[(\xi_0-\xi_3)(\xi_0\tau_2-\xi_3\tau_2)\hat{\pi}_{2}+\tau_3(\xi_0-\xi_3)^{2}\check{\pi}_{2}\big]\bigg|.
\end{eqnarray*}

Due to $\beta> 0, \theta=\exp2\beta>1$,  $\tau_1>0$,$\tau_3>0,
\xi_0>, \xi_3>0$, one can find that
 \begin{eqnarray*}
I_1=\frac{\tau_3\xi_3(2\tau_2+\tau_3)}{3\tau_3-2\tau_1}>0.
 \end{eqnarray*}
Now keeping in mind $0<\tau_1\leq \tau_3$ we have
\begin{eqnarray*}
\bigg|\frac{\tau_1}{\tau_3}-\frac{1}{2}\bigg|\leq \frac{1}{2}
\end{eqnarray*}
which yields
\begin{equation*}
I_2\left|\frac{\tau_1}{\tau_3}-\frac{1}{2}\right|^{n-1}\rightarrow
0 \ \ \textrm{as} \ \ n\rightarrow +\infty.
\end{equation*}
Then there exists $n_1\in \bn$ such that $\forall n\geq n_0$ one
has
\begin{equation*}
I_2\left|\frac{\tau_1}{\tau_3}-\frac{1}{2}\right|^{n}\leq
\frac{\varepsilon_1 }{2}.
\end{equation*}
Hence, for all $n\geq n_1$ we obtain
\begin{equation*}
\left|\varphi_1(E_{\Lambda_n})-\varphi_1(E_{\Lambda_n})\right|\geq
\frac{\varepsilon_1}{2}=\varepsilon_1\|E_{\Lambda_n}\|.
\end{equation*}
 This, according to Theorem \ref{br-q}, means that the states
 $\phi_1$ and $\phi_2$ are not quasi-equivalent.
The proof is complete.
\end{proof}

Now Theorems \ref{5.3},\ref{6.2} and \ref{6.9} imply Theorem
\ref{Main}.

%%%%%%%%%%%%%%%%%%%%%%%%%%%%%%%%%%%%%%%%%%%%%%%%%%%%%%%%%%%%%%%%%%%%%%%%%%%%%%%%%%%%%%%%%%%%%%%%%%%%%%%%
\section{Proof of Theorem \ref{main2}}
%%%%%%%%%%%%%%%%%%%%%%%%%%%%%%%%%%%%%%%%%%%%%%%%%%%%%%%%%%%%%%%%%%%%%%%%%%%%%%%%%%%%%%%%%%%%%%%%%%%%%%%

In this section we prove Theorem \ref{main2}. In this section we
assume that $\Delta(\theta)>0$ which ensures the existence of the
state $\varphi_1$ (see Theorem \ref{5.3}).

For each $n\in\bn$ let us define the following element:
\begin{equation}
 a_{\sigma}^{\Lambda_{n}}=\sigma^{(W_{n+1}(1))}.
\end{equation}
Clearly, one has $a_{\sigma}^{\Lambda_{n}}\in
\mathcal{B}_{\Lambda_{n}}$.

\begin{proposition}\label{7.1}
Let $\varphi_{\alpha}$ be the QMC associate to the pair
$(\omega_{\alpha},h_{\alpha})$.  Then one has
$$\varphi_{\alpha}(a_{\sigma}^{\Lambda_{n}})=0.$$
\end{proposition}

\begin{proof}
According to Proposition \ref{5.1} we have
\begin{eqnarray}\label{7.11}
\varphi_{\alpha}(a_{\sigma}^{\Lambda_{n}})
=\Tr\omega_{0}K_{[0,1]}\cdots
K_{[n-2,n-1]}Tr_{n-1]}\big[K_{[n-1,n]}a_{\sigma}^{\Lambda_{n}}\bh_{\alpha,
n}K_{[n-1,n]}^{*}\big] K_{[n-2,n-1]}^{*}\cdots K_{[0,1]}^{*}.
\end{eqnarray}
One can calculate that
\begin{eqnarray*}
\Tr_{n]}\big[K_{[n-1,n]}a_{\sigma}^{\Lambda_n}K_{[n-1,n]}^{*}\big]
&=&\Tr_{W_{n}(1)\big]}\big[A_{(x,(x,1),(x,2))}(\sigma h_{\alpha}^{(x,1)}\otimes h_{\alpha}^{(x,2)})A_{(x,(x,1),(x,2))}^{*}\big] \nonumber\\
&&\otimes\bigotimes_{x\in W_n\setminus W_n(1)}Tr_{x]}\big[A_{(x,(x,1),(x,2))}h_{\alpha}^{(x,1)}\otimes h_{\alpha}^{(x,2)})A_{(x,(x,1),(x,2))}^{*}\big]\nonumber\\
&=&\frac{1}{2}\tau_3\alpha\sigma^{W_{n-1}(1)}\bh_{\alpha, n-1}\nonumber\\
&=&\frac{1}{2}\tau_3\alpha a_{\sigma}^{\Lambda_{n-1}}\bh_{\alpha,
n-1}
\end{eqnarray*}
Therefore, from \eqref{7.11} with the last equality we obtain
\begin{eqnarray*}
\varphi_{\alpha}(a_{\sigma}^{\Lambda_{n}})&=&\frac{\tau_3\alpha}{2}\varphi_{\alpha}(a_{\sigma}^{\Lambda_{n-1}})\nonumber\\
&\vdots&\nonumber\\
&=&(\frac{\tau_3\alpha}{2})^{n}\varphi_{\alpha, 0}(a_{\sigma}^{\Lambda_{0}})\nonumber\\
&=&(\frac{\tau_3\alpha}{2})^{n}\Tr(\omega_{0}\sigma)=0.
\end{eqnarray*}
This completes the proof.
\end{proof}

Now using the same argument as in the proof of Proposition
\ref{6.6} one can prove the following

\begin{proposition}\label{7.2} For each $n\in\bn$ one has
\begin{eqnarray*}
\varphi_{1}(a_{\sigma}^{\Lambda_{n+1}})=\left[\left(\tau_1+\tau_2\right)\xi_0\xi_3\hat{\rho}_{1}+\frac{\tau_3}{2}\left(\xi_{0}^{2}+\xi_{3}^{2}\right)\check{\rho}_1\right]
+\left[\left(\tau_1+\tau_2\right)\xi_0\xi_3\hat{\rho}_{2}+\frac{\tau_3}{2}\left(\xi_{0}^{2}+\xi_{3}^{2}\right)\check{\rho}_2\right]\left(\frac{\tau_1}{\tau_3}-\frac{1}{2}\right)^{n}
\end{eqnarray*}
\end{proposition}

Now we are ready to prove Theorem \ref{main2}.

\begin{proof} First note that the state $\varphi_\a$ is a factor state, since it also satisfies the mixing property like the state $\varphi_1$.

Now from $0<\tau_1<\tau_3$ one finds
$|\frac{\tau_1}{\tau_3}-1|<\frac{1}{2}$ which yields
$$
\left(\frac{\tau_1}{\tau_3}-1\right)^{n}\to 0.$$

Due to Propositions \ref{7.1} and \ref{7.2} we get
\begin{eqnarray}\label{722}
|\varphi_{\alpha}(a_{\sigma}^{\Lambda_{n+1}})-{\varphi_{1}}(a_{\sigma}^{\Lambda_{n+1}})|
&=&|\varphi_{1}(a_{\sigma}^{\Lambda_{n+1}})|\nonumber\\
&=&|[(\tau_1+\tau_2)\xi_0\xi_3\hat{\rho}_{1}+\frac{1}{2}\tau_3(\xi_{0}^{2}+\xi_{3}^{2})\check{\rho}_1]\nonumber\\
&&+[(\tau_1+\tau_2)\xi_0\xi_3\hat{\rho}_{2}+\frac{1}{2}\tau_3(\xi_{0}^{2}+\xi_{3}^{2})\check{\rho}_2](\frac{\tau_1}{\tau_3}-\frac{1}{2})^{n}|\nonumber\\
&\geq&|(\tau_1+\tau_2)\xi_0\xi_3\hat{\rho}_{1}+\frac{1}{2}\tau_3(\xi_{0}^{2}+\xi_{3}^{2})\check{\rho}_1|\nonumber\\
&&-\bigg|(\tau_1+\tau_2)\xi_0\xi_3\hat{\rho}_{2}+\frac{1}{2}
\tau_3(\xi_{0}^{2}+\xi_{3}^{2})\check{\rho}_2\left(\frac{\tau_1}{\tau_3}-\frac{1}{2}\right)^{n}\bigg|,
\end{eqnarray}
here $\hat{\rho_1},\hat{\rho_2}$ and
$\check{\rho}_1,\check{\rho}_2$ are defined by
\eqref{r1},\eqref{r2}, respectively.

 We can check that
$\varepsilon_0:=\left[(\tau_1+\tau_2)\xi_0\xi_3\hat{\rho}_{1}+\frac{1}{2}\tau_3(\xi_{0}^{2}+\xi_{3}^{2})\check{\rho}_1\right]>0$.
Therefore, we have
\begin{eqnarray*}
|(\tau_1+\tau_2)\xi_0\xi_3\hat{\rho}_{2}+\frac{1}{2}\tau_3(\xi_{0}^{2}+\xi_{3}^{2})\check{\rho}_2
\left(\frac{\tau_1}{\tau_3}-\frac{1}{2}\right)^{n}|\rightarrow0.
\end{eqnarray*}
This means that there is $n_0\in\bn$ such that for all $n\geq n_0$
one gets
$$
\bigg|(\tau_1+\tau_2)\xi_0\xi_3\hat{\rho}_{2}+\frac{1}{2}\tau_3(\xi_{0}^{2}+\xi_{3}^{2})\check{\rho}_2\left(\frac{\tau_1}{\tau_3}-\frac{1}{2}\right)^{n}\bigg|\leq
\varepsilon_0/2.
$$
This due to  \eqref{722} implies
\begin{eqnarray*}
|\varphi_{\alpha}(a_{\sigma}^{\Lambda_{n+1}})-\varphi_{1}(a_{\sigma}^{\Lambda_{n+1}})|
\geq \frac{\varepsilon_{0}}{2}
\end{eqnarray*}
for all $n\geq n_0$.

For $\varepsilon=\frac{\varepsilon_{0}}{2}$, and
$\Lambda\subset_{fin} L$, there exists $n_1\in\bn$  such that
$\Lambda\subset \Lambda_{n_1}$
\begin{eqnarray*}
|\varphi_{\alpha}(a_{\sigma}^{\Lambda_{n+1}})-\varphi_{1}(a_{\sigma}^{\Lambda_{n+1}})|
\geq \varepsilon
\end{eqnarray*}
for all $n\geq \max\{n_0,n_1\}$. This from Theorem \ref{br-q} gets
the desired statement.
\end{proof}

\section*{Acknowledgments}
The authors are grateful to professors L. Accardi and F. Fidaleo
for their fruitful discussions and useful suggestions on the
definition of the phase transition. The authors also thank
referees whose valuable comments and remarks improved the
presentation of this paper.

\end{document}